\documentclass[onecolumn,letterpaper,11pt,draftclsnofoot]{IEEEtran}
\usepackage{epsfig,fullpage,amsmath,cite}
\usepackage{amssymb,multirow,delarray}
\usepackage{epsfig}
\usepackage{graphicx}
\usepackage{subfigure}

\newtheorem{theorem}{Theorem}

\newtheorem{definition}{Definition}

\usepackage{amssymb}
\usepackage{makeidx}
\usepackage{amsmath}
\usepackage{graphicx}
\usepackage{epsf}
\usepackage{epsfig}
\usepackage{psfig}
\usepackage{ccaption}
\usepackage{array}
\usepackage{tabularx}
\usepackage{multirow}
\usepackage{epsfig}
\usepackage{cite}

\begin{document}
%\topmargin= -  0.6in \oddsidemargin  -  0.5in\textwidth=7.5in
\textwidth 6.5in \textheight 9.1in \topmargin  -  0.1in

\title{\Large{Coalition Games with Cooperative Transmission:
A Cure for the Curse of Boundary Nodes in Selfish
Packet-Forwarding Wireless Networks}}
\author{Zhu Han$^*$ and H. Vincent Poor$^+$\\
$^*$Department of Electrical and Computer Engineering, \\Boise
State
University, Idaho, USA\\
$^+$Department of Electrical Engineering, \\ Princeton University,
New Jersey, USA  \thanks{This research was supported by the
National Science Foundation under Grants ANI-03-38807 and
CNS- 06-25637.}}

\maketitle\thispagestyle{empty}

\begin{abstract}

In wireless packet-forwarding networks with selfish nodes,
application of a repeated game can induce the nodes to forward
each others' packets, so that the network performance can be
improved. However, the nodes on the boundary of such networks
cannot benefit from this strategy, as the other nodes do not
depend on them. This problem is sometimes known as {\em the curse
of the boundary nodes}. To overcome this problem, an approach
based on coalition games is proposed, in which the boundary nodes
can use cooperative transmission to help the backbone nodes in the
middle of the network. In return, the backbone nodes are willing
to forward the boundary nodes' packets. Here, the concept of core
is used to study the stability of the coalitions in such games.
Then three types of fairness are investigated, namely, min-max
fairness using nucleolus, average fairness using the Shapley
function, and a newly proposed market fairness. Based on the
specific problem addressed in this paper, market fairness is a new
fairness concept involving fairness between multiple backbone
nodes and multiple boundary nodes. Finally, a protocol is designed
using both repeated games and coalition games. Simulation results
show how boundary nodes and backbone nodes form coalitions
according to different fairness criteria. The proposed protocol
can improve the network connectivity by about 50\%, compared with
pure repeated game schemes.

\end{abstract}

\newpage \setcounter{page}{1}
%  -   -   -   -   -   -   -   -   -   -   -   -   -   -   -   -   -   -   -   -   -   -   -   -   -   -   -   -   -   -   SECTION  -   -   -   -   -   -   -   -   -   -   -   -   -   -   -   -   -   -   -   -   -   -   -   -   -   -   -   -   -   -   -   -   -   -
\section{Introduction}\label{sec:intro}\setlength{\baselineskip}{25pt}

In wireless networks with selfish nodes such as ad hoc networks,
the nodes may not be willing to fully cooperate to accomplish the
overall network goals. Specifically  for the packet-forwarding
problem, forwarding of other nodes' packets consumes a node's
limited battery energy. Therefore, it may not be in a node's best
interest to forward other's arriving packets. However, refusal to
forward other's packets non-cooperatively will severely affect the
network functionality and thereby impair a node's own performance.
Hence, it is necessary to design a mechanism to enforce
cooperation for packet forwarding among greedy and distributed
nodes.

The packet-forwarding problem in ad hoc networks has been
extensively studied in the literature. The fact that nodes act
selfishly to optimize their own performance has motivated many
researchers to apply game theory \cite{Game_theory1,Game_theory2}
in solving this problem. Broadly speaking, the approaches used to
encourage packet-forwarding can be categorized into two general
types. The first type  makes use of virtual payments. Pricing
\cite{Crowfort_Gibbens_Kelly_Ostring02} and credit based method
\cite{Zhong_Chen_Yang03} fall into this first type. The second
type of approach is related to personal and community enforcement
to maintain the long-term relationship among nodes. Cooperation is
sustained because defection against one node causes personal
retaliation or sanction by others. \emph{Watchdog} and
\emph{pathrater} are proposed in \cite{Marti_Giuli_Kai_Baker00} to
identify misbehaving nodes and deflect traffic around them.
Reputation-based protocols are proposed in
\cite{Buchegger_LeBoudec02} and \cite{Michiardi_Molva03}. In
\cite{Altman_Kherani_Michiardi_Molva05}, a model is considered to
show cooperation among participating nodes. In \cite{Hubaux}, the
question of whether cooperation for packet forwarding can exist
without incentive mechanisms is answered using game theory and
graph theory. Packet forwarding schemes using ``TIT for TAT"
schemes are proposed in \cite{infocom03}. In \cite{hanzhu2}, a
cartel maintenance framework is constructed for distributed rate
control for wireless networks. In \cite{hanzhu1}, self-learning
repeated game approaches are constructed to enforce cooperation
and to study better cooperation. Some recent work applying game
theory to enhance energy-efficient behavior in infrastructure
networks can be found in
\cite{Meshkati1,Meshkati2,Meshkati3,Meshkati4}.

However, packet-forwarding networks are plagued by the so-called
{\em curse of the boundary nodes}. The nodes at the boundary of
the network must depend on the backbone nodes in the middle of the
networks to forward their packets. On the other hand, the backbone
nodes will not correspondingly depend on the boundary nodes. As a
result, the backbone nodes do not worry about retaliation or lost
reputation for not forwarding the packets of the boundary nodes.
This fact causes the curse of the boundary nodes. In order to cure
this curse, in this paper, we propose an approach based on
cooperative game coalitions using cooperative transmission.

Recently, cooperative transmission \cite{bib:Aazhang1}
\cite{bib:Laneman2} has gained considerable attention as a
transmit strategy for future wireless networks. The basic idea of
cooperative transmission is that relay nodes can help a source
node's transmission by relaying a replica of the source's
information. Cooperative communications efficiently takes
advantage of the broadcast nature of wireless networks, while
exploiting the inherent spatial and multiuser diversities. The
energy-efficient broadcast problem in wireless networks is
considered in \cite{Yates2}. The work in \cite{Luo} evaluates the
cooperative diversity performance when the best relay is chosen
according to the average signal-to-noise ratio (SNR), and the
outage probability of relay selection based on instantaneous SNRs.
In \cite{Bletsas}, the authors propose a distributed relay
selection scheme that requires limited network knowledge with
instantaneous SNRs. In \cite{bib:zhuAhmed}, the relay assignment
problem is solved for multiuser cooperative communications. In
\cite{bib:zhuwhohelpswhom}, cooperative resource allocation for
orthogonal frequency division multiplexing (OFDM) is studied. A
game theoretic approach for relay selection has been proposed in
\cite{globecom_zhu}. In \cite{hanzhu_CB}, the sensor nodes can
cooperate to form longer range communication links so as to bypass
the energy-depleting nodes. As a result, the network life time can
be greatly improved. In \cite{bib:Adve}, centralized power
allocation schemes are presented by assuming all the relay nodes
helped. In \cite{bib:Madsen}, cooperative routing protocols are
constructed based on non-cooperative routes.

Using cooperative transmission, boundary nodes can serve as relays
and provide some transmission benefits for the backbone nodes that
can be viewed as source nodes. In return, the boundary nodes are
rewarded for packet-forwarding. To analyze the benefits and
rewards, we investigate a game coalition that describes how much
collective payoff a set of nodes can gain and how to divide this
payoff. We investigate the stability and payoff division using
concepts such as the core, nucleolus, and Shapley function. Three
types of fairness are defined, namely, the min-max fairness using
nucleolus, average fairness using the Shapley function, and our
proposed market value fairness. Market fairness is a new fairness
concept involving multiple backbone nodes and multiple boundary
nodes, based on the specific problem treated in this paper. Then,
we construct a protocol using both repeated games and coalition
games. From the simulation results, we investigate how boundary
nodes and backbone nodes form coalitions according to different
fairness criteria. The proposed protocol can improve the network
connectivity by about 50\%, compared to the pure repeated game
approach.

This paper is organized as follows: In Section \ref{sec:model},
repeated game approaches are reviewed and the curse of the
boundary nodes is explained. In Section \ref{sec:protocol}, the
cooperative transmission model is illustrated and the
corresponding coalition games are constructed. Stability and three
types of fairness are investigated. A protocol that exploits the
properties of our approach is also proposed. Simulation results
are shown in Section \ref{sec:simulation} and conclusions are
given in Section \ref{sec:conclusion}.

%  -   -   -   -   -   -   -   -   -   -   -   -   -   -   -   -   -   -   -   -   -   -   -   -   -   -   -   -   -   -   SECTION  -   -   -   -   -   -   -   -   -   -   -   -   -   -   -   -   -   -   -   -   -   -   -   -   -   -   -   -   -   -   -   -   -   -
\section{Repeated Games and the Curse of Boundary Nodes\label{sec:model}}

A wireless packet-forwarding network can be modeled as a directed
graph $G(L,A)$, where $L$ is the set of all nodes and $A$ is the
set of all directed links $(i,l), i,l\in L$. Each node $i$ has
several transmission destinations which are included in set $D_i$.
To reach the destination $j$ in $D_i$, the available routes form a
{\em depending graph} $G_i^j$ whose nodes represents the potential
packet-forwarding nodes. The transmission from node $i$ to node
$j$ depends on a subsect of the nodes in $G_i^j$ for
packet-forwarding. Notice that this dependency can be mutual. One
node depends on the other node, while the other node can depend on
this node as well. In general, this mutual dependency is common,
especially for backbone nodes at the center of the network. In the
remainder of this section, we will discuss how to make use of this
mutual dependency for packet-forwarding using a repeated game, and
then we will explain the curse of boundary nodes.

\subsection{Repeated Games for Mutually Dependent Nodes}

A repeated game is a special type of dynamic game (a game that is
played multiple times). When the nodes interact by playing a
similar static game (which is played only once) numerous times,
the game is called a repeated game. Unlike a static game, a
repeated game allows a strategy to be contingent on the past
moves, thus allowing reputation effects and retribution, which
give possibilities for cooperation. The game is defined as
follows:

\begin{definition}
A $T$-period {\em repeated game} is a dynamic game in which, at
each period $t$, the moves during periods $1,\dots, t -  1$ are
known to every node. In such a game, the total discounted payoff
for each node is computed by $ \sum_{t=1}^T \beta^{t -  1} u_i(t),
$ where $u_i(t)$ denotes the payoff to node $i$ at period $t$ and
where $\beta$ is a discount factor. Note that $\beta$ represents
the node's patience or on the other hand how significantly the
past affects the current payoff. If $T = \infty$, the game is
referred as an infinitely-repeated game. The average payoff to
node $i$ is then given by:
\begin{equation}
    u_i=(1 -  \beta)\sum_{t=1}^{\infty}\beta ^{t -  1}u_i(t).
    \label{payoff}
\end{equation}
\end{definition}

It is known that repeated games can be used to induce greedy nodes
in communication networks to show cooperation. In
packet-forwarding networks, if a greedy node does not forward the
packets of other nodes, it can enjoy benefits such as power
saving. However, this node will be punished by the other nodes in
the future if it depends on the other nodes to forward its own
packets. The benefit of greediness in the short term will be
offset by the loss associated with punishment in the future. So
the nodes will rather act cooperatively if the nodes are
sufficiently patient. From the Folk theorem below, we infer that
in an infinitely repeated game, any feasible outcome that gives
each node a better payoff than the Nash equilibrium
\cite{Game_theory1,Game_theory2} can be obtained.

\begin{theorem}
({\em Folk Theorem} \cite{Game_theory1,Game_theory2}) Let
($\hat{u}_1, \dots , \hat{u}_L$) be the set of payoffs from a Nash
equilibrium and let ($u_1, \dots , u_L$) be any feasible set of
payoffs. There exists an equilibrium of the infinitely repeated
game that attains any feasible solution ($u_1, \dots , u_L$) with
$u_i \geq \hat{u}_i, \forall i$ as the average payoff, provided
that $\beta$ is sufficiently close to 1.
\end{theorem}

In the literature of packet-forwarding wireless networks, the
conclusion of the above Folk theorem is achieved by several
approaches. Tit-for-tat \cite{Altman_Kherani_Michiardi_Molva05}
\cite{infocom03} is proposed so that all mutually dependent nodes
have the same set of actions. A cartel maintenance scheme
\cite{hanzhu2} has closed-form optimal solutions for both
cooperation and non-cooperation. A self-learning repeated game
approach is proposed in \cite{hanzhu1} for individual distributed
nodes to study the cooperation points and to develop protocols for
maintaining them. Given the previous attention to the problem of
nodes having mutual dependency, we will assume in this paper that
the packet-forwarding problem of selfish nodes with mutual
dependency has been solved and we will focus instead on the
problems encountered by the boundary nodes.

%\begin{figure}[ht]
%\centerline{
%\begin{tabular}{cc}
%\psfig{system_model.eps, width=75truemm}&\psfig{figure=Game_model.eps,width=75truemm}\\
%{\small (a)}& {\small (b)}
%\end{tabular}
%}
%\caption{\label{example}{(a) Example of the Curse of Boundary
%Nodes (b) Coalition Game Model with Cooperative
%Transmission}}\vspace{ -  5mm}
%\end{figure}

%\begin{figure}\label{example}
%  \centering
%      \subfigure[Example of the Curse of Boundary
%Nodes] {
%        \includegraphics[width=65truemm]{system_model.eps}}
%        \hspace{10truemm}%
%        \subfigure[Coalition Game Model with Cooperative Transmission] {
%        \includegraphics[width=65truemm]{Game_model.eps}}
%    %  \caption{Multiple Relay Case and Performance Comparison}
%    % \label{fig: multiple}
%    \vspace{ -  10mm}
%\end{figure}

%\begin{figure}[htbp]
%\begin{center}
%    \epsfig{file=system_model.eps,width=80truemm}
%\end{center}
%\caption{Example of the Curse of Boundary
%Nodes}\label{curse_example}\vspace{ -  5mm}
%\end{figure}

\subsection{Curse of Boundary Nodes\label{curse}}

When there is no mutual dependency, the  curse of boundary nodes
occurs, an example of which is shown in Figure \ref{example}.
Suppose node $1$ needs to send data to node $3,$ and node $2$
needs to send data to node $0$. Because node $1$ and node $2$
depend on each other for packet-forwarding, they are obliged to do
so because of the possible threat or retaliation from the other
node. However, if node $0$ wants to transmit to node $2$ and node
$3$, or node $3$ tries to communicate with node $0$ and node $1$,
the nodes in the middle have no incentive to forward the packets
due to their greediness. Moreover, this greediness cannot be
punished in the future since the dependency is not mutual. This
problem is especially severe for the nodes on the boundary of the
network, so it is called {\em the curse of boundary nodes}.

On the other hand, if node $0$ can form a coalition with node $1$
and help node $1$'s transmission (for example to reduce the
transmitted power of node $1$), then node $1$ has an incentive to
help node $0$ transmit as a reward. A similar situation arises
 for node $3$ to form a coalition with node $2$. We call
 nodes like $1$ and $2$ {\it backbone nodes}, while nodes
like $0$ and $3$ are  {\it boundary nodes}. In the following
section, we will study how coalitions can be formed to address
this issue using cooperative transmission.

%  -   -   -   -   -   -   -   -   -   -   -   -   -   -   -   -   -   -   -   -   -   -   -   -   -   -   -   -   -   -   SECTION  -   -   -   -   -   -   -   -   -   -   -   -   -   -   -   -   -   -   -   -   -   -   -   -   -   -   -   -   -   -   -   -   -   -
\section{Coalition Games with Cooperative Transmission}\label{sec:protocol}

In this section, we first study a cooperative transmission
technique that allows nodes to participate in coalitions. Then, we
formulate a coalition game with cooperative transmission.
Furthermore, we investigate the fairness issue and propose three
types of fairness definitions. Finally, a protocol for
packet-forwarding using repeated games and coalition games is
constructed.

\subsection{Cooperative Transmission System Model}

First, we discuss the traditional direct transmission case. The
source transmits its information to the destination with power
$P_d$. The received SNR is
\begin{equation}\label{SNR_direct}
\Gamma_{d}=\frac {P_d|h_{s,d}|^2}{\sigma ^2},
\end{equation}
where $h_{s,d}$ is the channel response from the source to the
destination and $\sigma^2$ is the noise level. To achieve the
minimal link quality $\gamma$, we need for the transmitted power
to be sufficiently large so that $\Gamma_d \geq \gamma$. The
transmitted power is also upper bounded by $P_{max}$.

%\begin{figure}[htbp]
%\begin{center}
%    \epsfig{file=Game_model.eps,width=70truemm}
%\end{center}
%\caption{Coalition Game Model with Cooperative
%Transmission}\label{game_model}\vspace{ -  5mm}
%\end{figure}

Next, we consider  multiple nodes using the amplify-and-forward
protocol \cite{bib:Aazhang1}\footnote{Other cooperative
transmission protocols can be exploited in a similar way.} to
transmit in two stages as shown in Figure \ref{example_ct}. In
stage one, the source node (denoted as node $0$) transmits its
information to the destination, and due to the broadcast nature of
the wireless channels, the other nodes can receive the
information. In stage two, the remaining $N$ relay nodes help the
source by amplifying the source signal. In both stages, the source
and the relays transmit their signals through orthogonal channels
using schemes like TDMA, FDMA, or orthogonal CDMA.

In stage one, the source transmits its information, and the
received signals at the destination and the relays can be written
respectively as
\begin{equation}
y_{s,d}=\sqrt{P_0}h_{s,d}x+n_{s,d},
\end{equation}
\begin{equation}
\mbox{and  }y_{s,r_i}=\sqrt{P_0}h_{s,r_i}x+n_{s,r_i}, \forall i\in
\{1,\dots, N\},
\end{equation}
where $P_0$ is the transmitted power of the source, $x$ is the
transmitted symbol with unit power, $h_{s,r_i}$ is the channel
gain from the source to relay $i$, and $n_{s,d}$ and $n_{s,r_i}$
are noise processes at the destination and relay, respectively.
Without significant loss of generality, we assume that all noises
have the same power $\sigma ^2$.

In stage two, each relay amplifies the received signal from the
source and retransmits it to the destination. The received signal at
the destination for relay $i$ can be written as
\begin{equation}
y_{r_i,d}=\frac{\sqrt{P_i}}{\sqrt{P_0|h_{s,r_i}|^2+\sigma^2}}
h_{r_i,d}y_{s,r_i}+n_{r_i,d},
\end{equation}
where $P_i$ is relay $i$'s transmit power, $h_{r_i,d}$ is the
channel gain from relay $i$ to the destination,  and $n_{r_i,d}$
is noise with variance $\sigma ^2$.

At the destination, the signal received at stage one and the $N$
signals received at stage two are combined using maximal ratio
combining (MRC). The SNR at the output of MRC is
\begin{equation}\label{SNR_coop}
\Gamma=\Gamma_0+\sum _{i=1}^N \Gamma _i,
\end{equation}
where $\Gamma_0=\frac{P_0|h_{s,d}|^2}{\sigma ^2}$ and
\begin{equation}
\Gamma _i =\frac {P_0P_i|h_{s,r_i}|^2|h_{r_i,d}|^2} {\sigma^2
(P_0|h_{s,r_i}|^2+P_i|h_{r_i,d}|^2+\sigma^2)}.
\end{equation}

On comparing (\ref{SNR_coop}) with (\ref{SNR_direct}), in order to
achieve the desired link quality $\gamma$, we can see that the
required power is always less than the direct transmission power,
i.e., $P_0< P_d$. So cooperation transmission can reduce the
transmit power of the source node. This fact can give incentives
of mutual benefits for the backbone nodes (acting as sources) and
the boundary nodes (acting as relays), and consequently can cure
the curse of the boundary nodes mentioned in Section \ref{curse}.

\subsection{Coalition Game Formation for Boundary Nodes}

In this subsection, we study possible coalitions between the
boundary nodes and the backbone nodes, for situations in which the
boundary nodes can help relay the information of the backbone
nodes using cooperative transmission. In the following, we first
define some basic concepts that will be needed in our analysis.

\begin{definition}
A {\em coalition} $S$ is defined to be a subset of the total set of
nodes $\mathbb{N}=\{0,\dots , N\}$. The nodes in a coalition want
to cooperate with each other. The {\em coalition form} of a game
is given by the pair $(\mathbb{N},v)$, where $v$ is a real-valued
function, called  the {\em characteristic function}. $v(S)$ is the
value of the cooperation for coalition $S$ with the following
properties:
\begin{enumerate}
\item $v(\emptyset)=0$.

\item Super-additivity: if $S$ and $Z$ are disjoint coalitions
($S\bigcap Z=\emptyset$), then $v(S)+v(Z)\leq v(S\bigcup Z)$.
\end{enumerate}
\end{definition}

The coalition states the benefit obtained from cooperation
agreements. But we still need to examine whether or not the nodes
are willing to participate in the coalition. A coalition is called
{\it stable} if no other coalition will have the incentive and
power to upset the cooperative agreement. Such division of $v$ is
called a point in the {\em core}, which is defined by the
following definitions.

\begin{definition}
A payoff vector $\textbf U=(U_0,\dots, U_{N})$ is said to be {\em
group rational} or {\em efficient} if $\sum_{i=0}^{N}U_i=
v(\mathbb{N})$. A payoff vector $\textbf U$ is said to be {\em
individually rational} if the node can obtain the benefit no less
than acting alone, i.e. $U_i\geq v(\{i\}),\ \forall i$. An {\em
imputation} is a payoff vector satisfying the above two
conditions.
\end{definition}

\begin{definition}
An imputation $\textbf U$ is said to be unstable through a
coalition $S$ if $v(S)>\sum_{i\in S}U_i$, i.e., the nodes have
incentive for coalition $S$ and upset the proposed $\textbf U$.
The set $C$ of a stable imputation is called the {\em core}, i.e.,
\begin{equation}
C=\{\textbf U:\sum_{i\in \mathbb{N}}U_i=v(\mathbb{N}) \mbox{ and
}\sum_{i\in S}U_i\geq v(S),\ \forall S\subset \mathbb{N}\}.
\end{equation}
\end{definition}

%There are two types of coalitions with transferable utility or
%with untransferable utility. The transferable utility is defined
%as:
%\begin{definition}
%Transferable utility is the payoff that can be transferred between
%nodes, like money. For coalition games with transferable utility,
%one node can make a transfer to another so as to get them to
%participate in a coalition. These side payments facilitate
%coalition formation.
%\end{definition}
%For our case, the transferable utility between backbone nodes and
%boundary nodes is equivalent to energy.

In the economics literature, the core gives a reasonable set of
possible shares. A combination of shares is in the core if there
is no sub-coalition in which its members may gain a higher total
outcome than the combination of shares of concern. If a share is
not in the core, some members may be frustrated and may think of
leaving the whole group with some other members and form a smaller
group.

In the packet-forwarding network as shown in Figure
\ref{example_ct}, we first assume one backbone node to be the
source node (node $0$) and the nearby boundary nodes (node $1$ to
node $N$) to be the relay nodes. We will discuss the case of
multiple source nodes later. If no cooperative transmission is
employed, the utilities for the source node and the relay nodes
are
\begin{equation}
v(\{ 0\})= -  P_d,\mbox{ and } v(\{ i\})= -  \infty, \forall i=1,\dots,
N.
\end{equation}
Here a utility of $ -  \infty$ means that even though a boundary user tries to use
maximal power for transmission, it cannot successfully deliver any
packets due to the curse.

With cooperative transmission and a grand coalition that includes
all nodes, the utilities for the source node and the relay nodes
are
\begin{equation}\label{alpha}
U_0= -  P_0 -  \sum_{i=1}^N \alpha_i P_d
\end{equation}
\begin{equation}\label{Ui}
\mbox{and }U_i= -  \frac{P_i}{\alpha_i},
\end{equation}
where $\alpha_i$ is the ratio of the number of packets that the
backbone node is willing to forward for boundary node $i$, to the
number of packets that the boundary node $i$ relays for the
backbone node using cooperative transmission. Smaller $\alpha_i$
means the boundary nodes have to relay more packets before
realizing the rewards of packet forwarding. The other
interpretation of the utility is as the average power per
transmission for the boundary nodes\footnote{Notice that we omit
the transmitted power needed to send the boundary node's own
packet to the backbone node, since it is irrelevant to the
coalition.}. The following theorem gives conditions under which
the core is not empty, i.e, in which the grand coalition is
stable.

\begin{theorem}
The core is not empty if $\alpha_i\geq 0,\ i=1,\dots, N$, and
$\alpha_i$ are such that $U_0\geq v(\{ 0\})$, i.e,
\begin{equation}\label{condition_core}
\sum_{i=1}^N \alpha_i \leq \frac{P_d -  P_0}{P_d}.
\end{equation}
\end{theorem}
\begin{proof}
First, any relay node will get $ -  \infty$ utility if it leaves
the coalition with the source node, so no node has incentive to
leave the coalition with node $0$. Then, from (\ref{SNR_coop}),
the inclusion of relay nodes will increase the received SNR
monotonically. So $P_0$ will decrease monotonically with the
addition of any relay node. As a result, the source node has an
incentive to include all the relay nodes, as long as the source
power can be reduced, i.e., $U_0\geq v(\{ 0\})$. A grand coalition
is formed and the core is not empty if (\ref{condition_core})
holds.
\end{proof}

The concept of the core defines the stability of a utility
allocation. However, it does not define how to allocate the
utility. For the proposed game, each relay node can obtain
different utilities by using different values of  $\alpha_i$. In
the next three subsections, we study how to achieve min-max
fairness, average fairness, and market fairness.

\subsection{Min-Max Fairness of a Game Coalition using Nucleolus}

We introduce the concepts of {\em excess}, {\em kernel}, and {\em
nucleolus}\cite{Game_theory1,Game_theory2}. For a fixed
characteristic function $v$, an imputation $\textbf U$ is found
such that, for each coalition $S$ and its associated
dissatisfaction, an optimal imputation is calculated to minimize
the maximum dissatisfaction. The dissatisfaction is quantified as
follows.
\begin{definition}
The measure of dissatisfaction of an imputation $\textbf U$
for a coalition $S$ is defined as the {\em excess}:
\begin{equation}
e(\textbf U,S)=v(S) -  \sum_{j\in S}U_j.
\end{equation}
\end{definition}
Obviously, any imputation $\textbf U$ is in the core, if and only
if all its excesses are negative or zero.

\begin{definition}
A {\em kernel} of $v$ is the set of all allocations $\textbf U$
such that
\begin{equation}
\max_{S \subseteq \mathbb{N} -  j,i\in S}e(\textbf U,S)=\max_{T
\subseteq \mathbb{N} -  i,j\in T}e(\textbf U,T).
\end{equation}
If nodes $i$ and $j$ are in the same coalition, then the highest
excess that $i$ can make in a coalition without $j$ is equal to
the highest excess that $j$ can make in a coalition without $i$.
\end{definition}

\begin{definition}
The {\em nucleolus} of a game is the allocation $\textbf U$ that
minimizes the maximum excess:
\begin{equation}
    \textbf U=\arg \min _{\textbf U} (\max\ e(\textbf U,S),\ \forall S).
\end{equation}
\end{definition}

The nucleolus of a game has the following property: The nucleolus
of a game in coalitional form exists and is unique. The nucleolus
is group rational and individually rational. If the core is not
empty, the nucleolus is in the core and kernel. In other word, the
nucleolus is the best allocation under the min-max criterion.

Using the above concepts, we prove the following theorem to show
the optimal $\alpha_i$ in (\ref{alpha}) to have min-max fairness.
\begin{theorem}
The maximal $\alpha_i$ to yield the nucleolus of the proposed
coalition game is given by
\begin{equation}\label{minmax_solution}
\alpha_i=\frac{P_d -  P_0(\mathbb{N})}{NP_d},
\end{equation}
where $P_0(\mathbb{N})$ is the required transmitted power of the
source when all relays transmit with transmitted power $P_{max}$.
\end{theorem}
\begin{proof}
Since for any coalition other than the grand coalition, the excess
will be $ -  \infty$, we need only consider the grand coalition.
Suppose the min-max utility is $\mu$ for all nodes, i.e.
\begin{equation}
\mu= -  \frac {P_i}{\alpha_i}.
\end{equation}
From (\ref{condition_core}) and since $U_i$ is monotonically
increasing with $\alpha_i$ in (\ref{Ui}), we have
\begin{equation}
{\alpha_i} = \frac{P_i}{\sum_{i=1}^N{P_i}}\cdot
\frac{(P_d -  P_0)}{P_d}.
\end{equation}
Since $P_0$ in (\ref{SNR_coop}) is a monotonically increasing
function of $P_i$, to achieve the maximal $\alpha_i$ and $\mu$,
each relay transmits with the largest possible power $P_{max}$.
Notice here we assume the backbone node can accept arbitrarily
small power gain to join the coalition.
\end{proof}

\subsection{Average Fairness of Game Coalition using the Shapley Function}

The core concept defines the stability of an allocation of payoff
and the nucleolus concept quantifies the min-max fairness of a
game coalition. In this subsection, we study another average
measure of fairness for each individual using the concept of a
Shapley function \cite{Game_theory1,Game_theory2}.

\begin{definition}
A {\em Shapley function} $\phi$ is a function that assigns to each
possible characteristic function $v$ a vector of real numbers,
i.e.,
\begin{equation}
    \phi (v)=(\phi_0(v),\phi_1(v),\phi_2(v),\dots, \phi_N(v))
\end{equation}
where $\phi_i(v)$ represents the worth or value of node $i$ in the
game. There are four Shapley Axioms that $\phi(v)$ must satisfy
\begin{enumerate}
\item {\em Efficiency Axiom}: $\sum_{i\in
\mathbb{N}}\phi_i(v)=v(\mathbb{N})$.

\item {\em Symmetry Axiom}: If node $i$ and node $j$ are such that
$v(S\bigcup \{i\})=v(S\bigcup \{j\})$ for every coalition $S$ not
containing node $i$ and node $j$, then $\phi_i(v)=\phi_j(v)$.

\item {\em Dummy Axiom}: If node $i$ is such that $v(S)=v(S\bigcup
\{ i\})$ for every coalition $S$ not containing $i$, then
$\phi_i(v)=0$.

\item {\em Additivity Axiom}: If $u$ and $v$ are characteristic
functions, then $\phi(u+v)=\phi(v+u)=\phi(u)+\phi(v)$.

\end{enumerate}

It can be proved that there exists a unique function $\phi$
satisfying the Shapley axioms. Moreover, the Shapley function can
be calculated as
\begin{equation}
\phi_i(v)=\sum_{S\subset
\mathbb{N} -  i}\frac{(|S|)!(N -  |S|)!}{(N+1)!}[v(S\cup \{ i\}) -  v(S)].
\end{equation}
Here $|S|$ denotes the size of set $S$ and $\mathbb{N}=\{0, 1,
\dots N\}$.

\end{definition}

The physical meaning of the Shapley function can be interpreted as
follows. Suppose one backbone node plus $N$ boundary nodes form a
coalition. The nodes join the coalition in random order. So there
are $(N+1)!$ different ways that the nodes might be ordered in
joining the coalition. For any set $S$ that does not contain node
$i$, there are $|S|!(N -  |S|)!$ different ways to order the nodes
so that $S$ is the set of nodes that enter the coalition before
node $i$. Thus, if the various orderings are equally likely,
$|S|!(N - |S|)!/(N+1)!$ is the probability that, when node $i$
enters the coalition, the coalition of $S$ is already formed. When
node $i$ finds $S$ ahead of it as it joins the coalition, then its
marginal contribution to the worth of the coalition is $v(S\cup \{
i\}) - v(S)$. Thus, under the assumption of randomly-ordered
joining, the Shapley function of each node is its expected
marginal contribution when it joins the coalition.

In our specific case, we consider the case in which the backbone
node is always in the coalition, and the boundary nodes randomly
join the coalition. We have $v(\{0\})= -  P_d$ and
\begin{equation}
    v(\mathbb{N})=P_d -  P_0(\mathbb{N}) -  \sum_{i\in \mathbb{N}} \alpha_i P_d,
\end{equation}
which is the overall power saving. The problem here is how to find
a given node's $\alpha_i$ that satisfies the average fairness,
which is addressed by the following theorem.

\begin{theorem}
The maximal $\alpha_i$ that satisfies the average fairness with
the physical meaning of the Shapley function is given by
\begin{equation}\label{shapley_solution}
\alpha_i=\frac{P_i^s}{P_d},
\end{equation}
where $P_i^s$ is the average power saving with random entering
orders, which is defined as
\begin{equation}
P_i^s=\frac{1}{N}[P_d -  P_0(\{i\})] +\frac{\sum_{j=1,j\neq i}^N
[P_0(\{j\}) -   P_0(\{i,j\})]}{N(N -  1)} +\cdots.
\end{equation}
\end{theorem}

\begin{proof}
The maximal $\alpha_i$ is solved by the following equations:
\begin{equation}\label{Shapley_eqn}
\left\{
\begin{array}{l}
\frac{\alpha_i}{\alpha_j}=\frac{\phi_i}{\phi_j},\\
v(\mathbb{N})\geq 0.\\
\end{array}
\right.
\end{equation}
The first equation in (\ref{Shapley_eqn}) is the average fairness
according to the Shapley function, and the second equation in
(\ref{Shapley_eqn}) is the condition for a non-empty core. Similar
to min-max fairness, we assume that the backbone node can accept
arbitrarily small power gain to join the coalition.

If boundary node $i$ is the first to join the coalition, the
marginal contribution for power saving is $
\frac{1}{N}[P_d -  P_0(\{i\}) -  \alpha_i P_d]$, where $\frac 1 N$ is
the probability. If boundary node $i$ is the second to join the
coalition, the marginal contribution is $\frac{\sum_{j=1,j\neq
i}^N [P_0(\{j\})+\alpha_j P_d -   P_0(\{i,j\}) -   (\alpha_i+\alpha_j)
P_d]}{N(N -  1)}$. By means of some simple derivations, we can obtain
the Shapley function $\phi_i$ as
\begin{eqnarray}\label{phii}
\phi_i=  -  \alpha_i P_d+ \frac{1}{N}[P_d -  P_0(\{i\})]
+\frac{\sum_{j=1,j\neq i}^N [P_0(\{j\}) -   P_0(\{i,j\})]}{N(N -  1)}
+\cdots,
\end{eqnarray}
and then we can obtain
\begin{equation}
\alpha_i=\frac{[P_d -  P_0(\mathbb{N})]P_i^s}{P_d\sum_{j=1}^N P_j^s}.
\end{equation}
Since
\begin{equation}
P_d -  P_0(\mathbb{N})=\sum_{j=1}^N P_j^s,
\end{equation}
we prove (\ref{shapley_solution}).
\end{proof}

Notice that different nodes have different values of $P_i^s$, due
to their channel conditions and abilities to reduce the backbone
node's power. Compared with the min-max fairness in the previous
subsection, the average fairness using the Shapley function gives
different nodes different values of $\alpha_i$ according to their
locations.

%\begin{eqnarray}\label{phi0}
%\phi_0\approx  -  \frac{P_d}{N+1} -  \frac{\sum_{i=1}^N i C_0}{N+1}=
% -  \frac{P_d}{N+1} -  \frac{NC_0}{2}\approx  -  \frac{NC_0}{2}
%\end{eqnarray}
%where the first term in the right hand side (RHS) is the utility
%without any boundary node, the second term in RHS represents the
%probability times the marginal contribution when the backbone node
%is the second node to join the coalition, the third term in RHS is
%for the case where two boundary nodes are already in the
%coalition, and so on until the size of $S$ is equal to $N$.
%
%If boundary node $i$ enters the coalition with $K$ nodes and finds
%the backbone node is not in the coalition, the marginal
%contribution is zero. If the boundary node is the first to join
%the coalition, the contribution is also zero. The Shapley values
%for the boundary nodes are given by
%\begin{eqnarray}\label{phii}
%\phi_i=  -  \alpha_i P_d+ \frac{1}{N}[P_d -  P_0(\{i\})]
%+\frac{\sum_{j=1,j\neq i}^N [P_0(\{j\}) -   P_0(\{i,j\})]}{N(N -  1)}
%+\cdots.
%\end{eqnarray}
%Here the first term in RHS is the marginal contribution when the
%first node to join the coalition is the backbone node and the
%second is node $i$. The second term in RHS is when $K=2$ and
%$\frac 2 {N+1}$ is the probability that the backbone node join the
%coalition in the first two rounds. In the above Shapley value
%derivation, we have the assumption that the core not empty.

%From (\ref{phi0}) and (\ref{phii}),
%
%\begin{equation}
%\frac {\phi_0}{\phi_i}=\frac{ -  P_0 -  \sum_{i=1}^N \alpha_i P_d}{\frac
%{P_{max}}{\alpha_i}}, \forall i.
%\end{equation}

\subsection{Market Fairness of Game Coalition with Multiple Backbone Nodes}

In the previous two subsections, we have discussed two types of
fairness with one backbone node and multiple boundary nodes.
However, since the boundary nodes depend entirely on the backbone
node for packet forwarding, the backbone node can disregard the
fairness and coerce the boundary nodes by asking for an arbitrary
amount of payoff before helping the boundary nodes send their
packets. The backbone node can join the coalition only if
$v(\mathbb{N})>v_0$, where $v_0$ is a positive value. So the
$\alpha_i$ in (\ref{minmax_solution}) and (\ref{shapley_solution})
becomes
\begin{equation}
\alpha_i=\left\{
\begin{array}{ll}
\frac{P_d -  P_0(\mathbb{N}) -  v_0}{NP_d}, & \mbox{min-max fairness}\\
\frac{[P_d -  P_0(\mathbb{N}) -  v_0]P_i^s}{P_d[P_d -  P_0(\mathbb{N})]}, &
\mbox{average fairness}.
\end{array}
\right.
\end{equation}

Thus, a greedy backbone node can increase $v_0$ sufficiently large
so that the boundary nodes receive arbitrarily small $\alpha_i$.
This means that the backbone node can arbitrarily impose on the
boundary nodes for relaying and almost never give rewards in
return. The underlying reason for this is because the backbone
node has no competition from other nodes. In economic networks,
this phenomena is called ``monopoly" and the consumer suffers a
minimal quality of services as a result.

To solve this problem, we discuss the case in which there are
multiple backbone nodes, which is similar to ``antitrust" in
economic networks. First, we prove the following theorem for the
core with multiple backbone nodes.
\begin{theorem}
If the number of backbone nodes is greater than $1$, the core is
surely empty.
\end{theorem}
\begin{proof}
Suppose there are $M$ backbone nodes. There is no mutual benefit
between these backbone nodes to form coalitions using cooperative
transmission. The boundary nodes join the coalitions that give
them the highest payoff, i.e., they seek coalitions that allow
them to relay the fewest packets before a reward for packet
forwarding is given. So the grand coalition is divided into $M$
coalitions, since there are benefits for a subset of nodes to form
the new coalition instead of joining the grand coalition. As a
result, the core is surely empty when $M\geq 2$.
\end{proof}

Since the grand coalition does not exist and the boundary nodes
can select the backbone nodes with which to form coalitions, it is
to the backbones nodes' benefits to adjust the packet-forwarding
policy to attract more boundary nodes to reduce the transmitted
power. So competition among the backbone nodes is introduced. Here
we denote $\alpha _ i ^m$ as the $m^{th}$ backbone node's
packet-forwarding policy for the $i^{th}$ boundary node, and
denote the nodes' coalition partition matrix as
\begin{equation}
\textbf A_{im}=\left \{
\begin{array}{ll}
1, & \mbox{ if boundary node }i\mbox{ joins coalition with
backbone node }
m,\\
0, & \mbox{ otherwise}.
\end{array}
\right.
\end{equation}
Boundary node $i$ selects the smallest $\alpha_i^m$ and joins the
corresponding coalition with backbone node $m$. The optimal
utility is obtained when the packets are relayed only by this
backbone node. So we have
\begin{equation}
\sum _{m=1}^M \textbf {A}_{im}=1, \textbf {A}_{im}\in \{0,
1\},\forall i,m.
\end{equation}

For each backbone node, the optimization is to adjust its policy
$\alpha _i ^m$ so that the overall power saving is largest. We can
write the utility for the $m^{th}$ backbone node as
\begin{equation}\label{game_backbone}
U_0^m=\max_{\alpha_i^m, \forall i} ( -   P_0^m -  \sum _{i=1}^N \alpha
_i^m \textbf A _{im} P_d),
\end{equation}
where $P_0^m$ is the reduced transmitted power using cooperative
transmission.

For each boundary node, the optimization is to select the backbone
node to join the coalition. The problem can be written as
\begin{equation}\label{game_boundary}
U_i=\max _{\textbf A _{im}, \forall m}  -  \frac {P_i}{\sum _{m=1}^M
\textbf A _{im} \alpha _i^m}
\end{equation}
\[
\mbox{s.t. } \sum _{m=1}^M \textbf {A}_{im}=1, \textbf {A}_{im}\in
\{0, 1\},\forall m.
\]

In order for a backbone node to win the coalition with a boundary
node, the backbone node has to set the lowest $\alpha_i^m$ among
all the backbone nodes. On the other hand, different boundary
nodes have different abilities to reduce different backbone nodes'
power. Using these facts, we define a new type of fairness as
follows.
\begin{definition}
{\em Market fairness} achieves the equilibrium of the two-level
games in (\ref{game_backbone}) and (\ref{game_boundary}). In this
type of fairness, no backbone node can set its policy $\alpha_i^m$
lower to get a higher utility if the other backbone nodes do not
change their policies. On the other hand, it is in the boundary
nodes' best interest to join the coalitions under this market
fairness.
\end{definition}

In Figure \ref{example_mf}, we show an example of market fairness
for a two-backbone-node and one-boundary-node example. The x-axis
and y-axis are the two backbone nodes' policies $\alpha_1^1$ and
$\alpha_1^2$, respectively. Below the $45$ degree line,
$\alpha_1^1 < \alpha _1 ^2$. As a result, the boundary node joins
the coalition with backbone node $1$, i.e., $\textbf A _{11}=1$
and $\textbf A _{12}=0$. Otherwise, the boundary node forms a
coalition with backbone node $2$, i.e., $\textbf A _{11}=0$ and
$\textbf A _{12}=1$. From (\ref{Ui}), the utility is a linear
function of the policy $\alpha _1^m,\ m=1,2$. For a backbone node,
the minimal requirement for joining the coalition is that its
power is less than the direct transmission power. In Figure
\ref{example_mf}, we show the $\tilde \alpha _1^m$ for $U_0^m= -
P_d, m=1,2$. Since different boundary nodes can help reduce the
backbone nodes' transmission power $P_0^m$ differently, different
values of $\tilde \alpha_1^m$ are required to achieve $U_0^m= -
P_d$. In our case, backbone node $1$ is the winner by providing
$\alpha_1^1=\tilde \alpha_1^2$ and its utility gain is $P_d -
P_0^1 -  \tilde \alpha _1^2 P_d$. Notice that backbone node $1$
cannot let $\alpha_1^1<\tilde \alpha_1^2$, since backbone node $2$
then has the ability to give lower $\alpha_1^2$ to attract the
boundary node. On the other hand, as long as $\alpha_1^1>\tilde
\alpha_1^2$, backbone node $1$ has no incentive to increase
$\alpha_1^1$. From this example, we can see that the backbone
nodes have to offer high enough $\alpha_i^m$ to form coalitions
with the boundary nodes, because of the competition with other
backbone nodes.

Notice that the cooperative transmission power $P_0^m$ depends on
which boundary nodes join the coalition, i.e., $P_0^m$ is a
function of vector $[\textbf A_{1m}, \dots ,\textbf A_{Nm}]$. In
order to find the point with market fairness, we formulate the
following programming method.

\begin{equation}\label{mf_linear}
\min _{P_0^m,\alpha _i^m>0, \forall i,m} \sum_{m=1}^M
\left(P_0^m(\textbf A_{1m}, \dots, \textbf A_{Nm})+\sum_{i=1}^N
\alpha_i^m \textbf A _{im} P_d\right),
\end{equation}
\[
\mbox{s.t. }\left\{
\begin{array}{l}
\forall i, \textbf A_{im}=1, \mbox{ if } \alpha
_i^m>\alpha_i^{m'}\mbox{ and } \forall
m'\neq m; \textbf A_{im}=0, \mbox{ otherwise,}\\
\forall m, P_0^m(\textbf A_{1m}, \dots \textbf
A_{Nm})+\sum_{i=1}^N \alpha_i^m \textbf A _{im} P_d \leq P_d.
\end{array}
\right.
\]
Here the first constraint is {\em Boundary Node Rationality},
which means the boundary node will select the optimal
$\alpha_i^m$. The second constraint is {\em Backbone Node
Rationality}, which states that the backbone nodes will only join
a coalition if their power can be reduced by doing so. The
objective function captures the phenomenon that the backbone nodes
will maximize their utilities by reducing $\alpha_i^m$ as much as
possible. The problem in (\ref{mf_linear}) can be efficiently
solved by algorithms such as the {\em cutting-plane} and {\em
simplex} algorithms \cite{non_optimal,optimization2}. The problem
can be solved by either the backbone node or the boundary node,
since the outcome can benefit both nodes.

\subsection{Joint Repeated-Game and Coalition-Game Packet-Forwarding Protocol}

Using the above analysis, we now develop a packet-forwarding
protocol based on repeated games and coalition games having on the
following steps.

\begin{center}{{\em Packet-Forwarding Protocol using Repeated
Games and Coalition Games}}
\end{center}
\begin{enumerate}
\item  Route discovery for all nodes.

\item  Packet-Forwarding enforcement for the backbone nodes, using
threat of future punishment in the repeated games.

\item Neighbor discovery for the boundary nodes.

\item Coalition game formation.

\item Packet relay for the backbone nodes with cooperative
  transmission.

\item Transmission of the boundary nodes' own packets to the
backbone nodes for forwarding.
 \vspace{5mm}
\end{enumerate}

First, all nodes in the network undergo route discovery. Then each
node knows who depends on it and on whom it depends for
transmission. Using this route information, the repeated games can
be formulated for the backbone nodes. The backbone nodes forward
the other nodes' information because of the threat of future
punishment if these packets are not forwarded. Due to the network
topology, some nodes' transmissions depend on the others while the
others do not depend on these nodes. These nodes are most often
located at the boundary of the network. In the next step, these
boundary nodes try to find their neighboring backbone nodes. Then,
the boundary nodes try to form coalitions with the backbone nodes,
so that the boundary nodes can be rewarded for transmitting their
own packets. Cooperative transmission gives an opportunity for the
boundary nodes to pay some ``credits" first to the backbone nodes
for the rewards of packet-forwarding in return. On the other hand,
competition among the backbone nodes prevents the boundary nodes
from being forced to accept the minimal payoffs.

It is worth mentioning the following point regarding energy
efficiency. From the overall system point of view, it is not
energy efficient for the backbone nodes to depend on the boundary
nodes for cooperation, since the boundary nodes are further away
from the backbones' destination. If a centralized control system
is enforced, it is energy efficient for the backbone nodes to
forward the packets of the boundary nodes. However, if distributed
and greedy users are considered, the curse happens. Our approach
provides the incentives for the backbone nodes to help the
boundary nodes, so that the curse is relieved. But this comes with
a cost, in the sense that the boundary nodes have to help the
backbone nodes in an energy-inefficient way. Nevertheless, this is
already much better than the totally accursed situation in which
no packet of the boundary nodes can be transmitted.

Another implementation concern arises from node mobility. The
proposed algorithm is similar to a contract. As long as the
boundary nodes help the transmission and the backbone nodes help
the packet forwarding, the contract is fulfilled. This can happen
with transmission of $1+\frac 1 \alpha$ packets. If mobility is
considered, the new contract needs to be calculated to accounts
for the new positions of the nodes. As long as the speed of
fulfilling the contract is relatively larger than the speed of
channel changing, the proposed scheme can be implemented without
major modification. Nevertheless, if the channel changes too
rapidly, some stochastic models can be used to estimate the
expected payoff (utility). Then the rest of the analysis can be
applied in a very similar way.

%\begin{table}
%\caption{Packet-Forwarding Protocol}
%\begin{center}\label{protocol}
%\begin{tabular}{|l|}
%  \hline
%  % after \\: \hline or \cline{col1 -  col2} \cline{col3 -  col4} ...
%  1. Route discovery for all nodes.\\
%  \hline
%  2. Packet-Forwarding enforcement for the backbone nodes, using threat of future punishment in the repeated games. \\
%  \hline
%  3. Neighbor discovery for the boundary nodes.\\
%  \hline
%  4. Coalition game formation.\\
%  \hline
%  5. Relaying the backbone node's packets with cooperative
%  transmission.\\
%  \hline
%  6. Transmitting the boundary nodes' own information to the backbone nodes for
%  forwarding.\\
%  \hline
%\end{tabular}
%\end{center}
%\end{table}

%\begin{figure}[htbp]
%\begin{center}
%    \epsfig{file=a_R.eps,width=70truemm}
%\end{center}
%\caption{$\alpha$ for Different Channels and no.  of
%Nodes}\label{a_R}\vspace{ -  5mm}
%\end{figure}

%\begin{figure}[htbp]
%\begin{center}
%    \epsfig{file=prob_size.eps,width=70truemm}
%\end{center}
%\caption{Network Connectivity as a function of Network
%Size}\label{prob_size}\vspace{ -  5mm}
%\end{figure}
%  -   -   -   -   -   -   -   -   -   -   -   -   -   -   -   -   -   -   -   -   -   -   -   -   -   -   -   -   -   -   SECTION  -   -   -   -   -   -   -   -   -   -   -   -   -   -   -   -   -   -   -   -   -   -   -   -   -   -   -   -   -   -   -   -   -   -
\section{Simulation Results}\label{sec:simulation}

We model all channels as additive white Gaussian noise channels
having the exponent of propagation loss as $3$; that is, power
falls off spatially according to an inverse-cubic law. The maximal
transmitted power is $10$dbm and the noise level is
 -  $60$dbm. The minimal SNR $\gamma$ is 10dB. In the first setup, we
assume the backbone node is located at $(0m,0m)$, and the
destination is located at either $(100m,0m)$ or $(50m,0m)$. The
boundary nodes are located on an arc with angles randomly
distributed from $0.5\pi$ to $1.5\pi$ and with distances varying
from $5$m to $100$m.

In Figure \ref{a_R}, we study the min-max fairness and show the
average $\alpha_i$ over $1000$ iterations as a function of
distance from the relays to the source node. Due to the min-max
nature, all boundary nodes have the same $\alpha_i$. When the
distance is small, i.e., when the relays are located close to the
source, $\alpha_i$ approaches $\frac 1 N$. This is because the
relays can serve as a virtual antenna for the source, and the
source needs very low power for transmission to the relays. When
the distance is large, the relays are less effective and
$\alpha_i$ decreases, which means that the relays must transmit
more packets for the source to earn the rewards of
packet-forwarding. When the destination is located at $50$m, the
source-destination channel is better than that at 100$m$. When
$N=1$ and the source-destination distance is $50$m, the relays
close to the source have larger values of $\alpha_i$ and the
relays farther away have lower values of $\alpha_i$ than that in
the $100$m case. In Figure \ref{P0_R}, we show the corresponding
$P_0$ for the backbone node. We can see that $P_0$ increases when
the distances between the boundary nodes to the backbone node
increase.

If we consider the multiple backbone (multiple core) case  with
min-max fairness, Figure \ref{a_R} and Figure \ref{P0_R} provide
the boundary nodes a guideline for selecting a backbone node with
which to form a coalition. First, a less crowded coalition is
preferred. Second, the nearest backbone node is preferred. Third,
for  $N=1$, if the source-destination channel is good, the closer
backbone node is preferred; otherwise, the farther one can provide
larger $\alpha_i$.

%\begin{figure}
%  \centering
%      \subfigure[$\alpha$ for Different Channels and no.  of
%Nodes] {
%        \includegraphics[width=65truemm]{a_R.eps}}
%        \hspace{10truemm}%
%        \subfigure[Network Connectivity as a function of Network
%Size] {
%        \includegraphics[width=65truemm]{prob_size.eps}}
%    %  \caption{Multiple Relay Case and Performance Comparison}
%    % \label{fig: multiple}
%    \vspace{ -  10mm}\label{a_R}
%\end{figure}

Next, we investigate the average fairness using the Shapley
function. The simulation setup is as follows. The backbone node is
located at $(0m,0m)$ and the destination is located at
$( -  50m,0m)$. Boundary node $1$ is located at $(20m,0m)$ or
$(50m,0m)$. Boundary node $2$ moves from $(5m,0m)$ to $(100m,0m)$.
The remaining simulation parameters are the same. In Figure
\ref{a_shapley}, we show maximal $\alpha_i$ for two boundary
nodes. We can see that when boundary node $2$ is closer to the
backbone node than boundary node $1$, $\alpha_2>\alpha_1$, i.e.,
boundary node $2$ can help relay fewer packets for backbone node
$1$ before being rewarded. The two curves for $\alpha_1$ and
$\alpha_2$ for the same boundary node $1$ location cross at the
boundary node $1$ location. The figure shows that the average
fairness using the Shapley function gives greater rewards to the
boundary node whose channel is better and who can help the
backbone node more. When boundary node $2$ moves from $(20m,0m)$
to $(50m,0m)$, $\alpha_1$ becomes smaller, but $\alpha_2$ becomes
larger. This is because the backbone node must depend on boundary
node $2$ more for relaying. However, the backbone node will pay
less for the boundary nodes. Notice that $\alpha_i$ at the
crossover point is lower. This is because the overall power for
the backbone node is high when boundary node $2$ is far away, as
shown in Figure \ref{P0_R}.

Further, we study market fairness with the following setup.
Backbone node $1$ and backbone node $2$ are located at $(0m, -  30m)$
and $(0m,30m)$, respectively. The destination is located at
$( -  50m,0m)$ and boundary node $1$ is located at $(44m,10m)$.
Boundary node $2$ moves from $(44m, -  50m)$ to $(44m,50m)$. The
remaining simulation parameters are the same as before. In Figure
\ref{a_market} and Figure \ref{P0_market}, we show $\alpha_i^m$
and $P_0^m$, respectively, under six different scenarios:
\begin{enumerate}

\item (2,1): Coalition of boundary node $2$ with backbone node
$1$;

\item (2,2): Coalition of boundary node $2$ with backbone node
$2$;

\item (1,1): Coalition of boundary node $1$ with backbone node
$1$;

\item (1,2): Coalition of boundary node $1$ with backbone node
$2$;

\item ([1;2],1): Coalition of both boundary nodes with backbone
node $1$;

\item ([1;2],2): Coalition of both boundary nodes with backbone
node $2$.

\end{enumerate}
Since boundary node $1$ is not moving, coalition (1,1) and
coalition (1,2) are horizontal lines. From the curves in Figure
\ref{a_market}, we can see that a boundary node prefers to form a
coalition with the closest backbone node, and vice versa. However,
due to competition from the other nodes, the coalition formation
is affected by combinations of many factors which are analyzed as
follows.

When boundary node $2$ moves, there are seven possible scenarios
for forming different coalitions in Table \ref{market_table}. From
Figure \ref{a_market} and Table \ref{market_table}, we can see
that rational boundary node $i$ selects the largest $\alpha_i^m$
and joins the corresponding coalition with backbone node $m$.
Sometime, it is to both boundary nodes' benefit to form a
coalition with one of the backbone nodes as in case II and case
VI. However, because the backbone nodes are greedy, the boundary
nodes can obtain only slightly better rewards than the opponent's
offer. For example, in case I, boundary node $1$ prefers backbone
node $2$. But as long as the backbone node gives an offer better
than $\alpha_1^1$, boundary node $1$ must accept the offer. On the
other hand, from Figure \ref{P0_market}, the backbone nodes want
to form coalitions with both boundary nodes so as to have the
minimal transmitted power. But because of competition from other
backbone nodes and rationality of the boundary nodes, the backbone
nodes must form a coalition with only one boundary node or
sometime not at all. The above facts demonstrate the reason why
the proposed market fairness can effectively counteract the
greediness of the backbone nodes.

%\begin{enumerate}
%
%\item Boundary node $1$ form a coalition with backbone node $2$,
%and boundary node $2$ forms a coalition with backbone node $1$.
%
%\item Both boundary node $1$ and $2$ form a coalition with
%backbone node $1$.
%
%\item Boundary node $2$ form a coalition with backbone node $1$,
%and boundary node $1$ forms a coalition with backbone node $2$.
%
%\item Boundary node $1$ form a coalition with backbone node $2$,
%and boundary node $2$ forms a coalition with backbone node $1$.
%
%\item Boundary node $2$ form a coalition with backbone node $2$,
%and boundary node $1$ forms a coalition with backbone node $1$.
%
%\item Both boundary node $1$ and $2$ form a coalition with
%backbone node $2$.
%
%\item Boundary node $1$ form a coalition with backbone node $2$,
%and boundary node $2$ forms a coalition with backbone node $1$.
%
%\end{enumerate}
%

Next, we set up a linear network with $50$ nodes spread evenly
along a line. The distance between nodes is $100$m. The users are
indexed as user $1$ to user $50$ from one end to the other. Each
user transmits to any other user with equal probability. In Figure
\ref{linear_setup}, we show the probability that a node can be a
boundary node as a function of the user index. We show two cases
with one destination and five destinations for each user,
respectively. We can see that the nodes in the middle of the
network have lower probabilities to become boundary nodes, as one
would expect. As a result, $\alpha$ for those nodes is large,
which means those nodes take less on average to help the others'
transmission, because of their locations. For the five destination
case, each user transmits to five different destination nodes, and
as a result depends more on the other nodes. So the nodes in the
middle have much lower probability to be boundary nodes than in
the one-destination case.

Finally, we examine the degree to which the coalition game can
improve the network connectivity. Here we define the network
connectivity as the probability that a randomly located node can
connect to the other nodes. All nodes are randomly located within
a square of size $B\times B$. In Figure \ref{prob_size}, we show
the network un-connectivity as a function of $B$ for the numbers
of nodes equal to $100$ and $500$. With increasing network size,
the node density becomes lower, and more and more nodes are
located at the boundary and must depend on the others for
packet-forwarding. If no coalition game is formed, these boundary
nodes cannot transmit their packets due to the selfishness of the
other nodes. With the coalition game, the network connectivity can
be improved by about 50\%. The only chance that a node cannot
connect to the other nodes is when this node is located too far
away from any other node. Thus, we can seen by this example that
the coalition game can cure the curse of the boundary nodes in
wireless packet-forwarding networks with selfish nodes.

%  -   -   -   -   -   -   -   -   -   -   -   -   -   -   -   -   -   -   -   -   -   -   -   -   -   -   -   -   -   -   SECTION  -   -   -   -   -   -   -   -   -   -   -   -   -   -   -   -   -   -   -   -   -   -   -   -   -   -   -   -   -   -   -   -   -   -
\section{Conclusions}\label{sec:conclusion}

In this paper, we have proposed a coalition game approach to
provide incentives to selfish nodes in wireless packet-forwarding
networks using cooperative transmission, so that the boundary
nodes can transmit their packets effectively. We have used the
concepts of coalition games to maintain stable and fair game
coalitions. Specifically, we have studied three fairness concepts:
min-max fairness, average fairness, and market fairness. The
market fairness uses competition among nodes to effectively
counteract the greediness of the backbone nodes. A protocol has
been constructed using repeated games and coalition games. From
simulation results, we have seen how boundary nodes and backbone
nodes form coalitions according to different fairness criteria. We
have also seen by example that network connectivity can be
improved by about 50\%, compared to the pure repeated game
approach.

%  -   -   -   -   -   -   -   -   -   -   -   -   -   -   -   -   -   -   -   -   -   -   -   -   -   -   -   -   -   -   BIBLIOGRAPHY  -   -   -   -   -   -   -   -   -   -   -   -   -   -   -   -   -   -   -   -   -   -   -   -   -   -   -   -   -   -   -   -   -   -
\bibliographystyle{IEEE}

\begin{thebibliography}{1}\setlength{\baselineskip}{12pt}

\bibitem{Game_theory1} G. Owen,
\emph{Game Theory}, 3rd ed. Academic Press, Burlington, MA, 2001.

\bibitem{Game_theory2} R. B. Myerson,
\emph{Game Theory: Analysis of Conflict}, Harvard University
Press, Cambridge, MA, 1991.

%\bibitem{Buttyan_Hubaux03} L. Buttyan and J. -  P. Hubaux,
%``Stimulating cooperation in self -  organizing mobile ad hoc
%networks," {\em ACM/Kluwer MONET Special Issue on Mobile Ad Hoc
%Networks}, vol.  8, no.  5, October, 2003.

\bibitem{Crowfort_Gibbens_Kelly_Ostring02} J. Crowcroft, R.
Gibbens, F. Kelly and S. Ostring, ``Modelling incentives for
collaboration in mobile ad hoc networks," {\em Performance
Evaluation}, vol. 57, no. 4, pp. 427 -  439, August 2004.

\bibitem{Zhong_Chen_Yang03} S. Zhong, J. Chen and Y. R.
Yang, ``Sprite: A simple, cheat-proof, credit-based system for
mobile ad-hoc networks," in {\em Proceedings of the Annual IEEE
Conference on Computer Communications, INFOCOM}, pp. 1987 -  1997,
San Francisco, CA, March 2003.

\bibitem{Marti_Giuli_Kai_Baker00} S. Marti, T. J. Giuli, K. Lai and M.
Baker, ``Mitigating routing misbehaviour in mobile ad hoc
networks," in {\em Proceedings of the ACM/IEEE Annual International
Conference on Mobile Computing and Networking (Mobicom)},
pp. 255 -  265, Boston, MA, August 2000.

\bibitem{Buchegger_LeBoudec02} S. Buchegger and J -  Y. Le Boudec,
``Performance analysis of the CONFIDANT protocol (cooperation of
nodes -  fairness in dynamic ad-hoc networks)," in {\em
Proceedings of the ACM International Symposium on Mobile Ad Hoc
Networking and Computing (MobiHoc)}, pp. 80 -  91, Lausannae,
Switzerland, June 2002.

\bibitem{Michiardi_Molva03} P. Michiardi and R. Molva,
``A game theoretical approach to evaluate cooperation enforcement
mechanisms in mobile ad hoc networks," in {\em Proceedings of the
IEEE/ACM International Symposium on Modeling and Optimization in
Mobile, Ad Hoc, and Wireless Networks (WiOpt)}, Sophi Antipolis,
France, March 2003.

%\bibitem{Felegyhazi_Buttyan_Hubaux03a} M. Felegyhazi, L. Buttyan, and J. P.
%Hubaux, ``Equilibrium analysis of packet forwarding stratiegies in
%wireless ad hoc networks  -   the static case," {\em Proceedings of the
%PWC}, Venice, Italy, 2003.

\bibitem{Altman_Kherani_Michiardi_Molva05} E. Altman, A. A. Kherani, P. Michiardi and R. Molva,
``Non-cooperative forwarding in ad hoc networks," in {\em
Proceedings of the {I}nternational {C}onferences on {N}etworking},
{v}ol.3462, {W}aterloo, {C}anada, May 2005.

\bibitem{Hubaux} M. Felegyhazi, J. P. Hubaux, and L. Buttyan,
``Nash equilibria of packet forwarding strategies in wireless ad
hoc networks," {\em IEEE Transactions on Mobile Computing} , vol.
5, no.  5, p.p. 463 -  476, April 2006.

\bibitem{infocom03} V. Srinivasan, P. Nuggehalli, C. F. Chiasserini, and R. R.
Rao, ``Cooperation in wireless ad hoc networks," in \emph{
Proceedings of the Annual IEEE Conference on Computer Communications
(INFOCOM)}, San Francisco, CA, March 2003.

\bibitem{hanzhu2}Z. Han, Z. Ji, and K. J. R. Liu,
``Dynamic distributed rate control for wireless networks by
optimal cartel maintenance strategy," in {\em Proceedings of the  the IEEE
Global Telecommunications Conference}, pp. 3742 -  3747, Dallas, TX,
November 2004.

\bibitem{hanzhu1}Z. Han, C. Pandana, and K. J. R. Liu,
``A self-learning repeated game framework for optimizing packet
forwarding networks," in {\em Proceedings of the IEEE Wireless
Communications and Networking Conference}, pp. 2131 -  2136, New
Orleans, LA, March 2005.




\bibitem{Meshkati1} F. Meshkati, H. V. Poor, S. C. Schwartz and N. B. Mandayam, ``An
energy-efficient approach to power control and receiver design in
wireless data networks," {\em IEEE Transactions on
Communications}, vol.  53, no.  11, pp.  1885  -   1894, November
2005.


\bibitem{Meshkati2}F. Meshkati, M. Chiang, H. V. Poor and S. C. Schwartz, ``A
game-theoretic approach to energy -  efficient power control in
multi-carrier CDMA systems," {\em IEEE Journal on Selected Areas
in Communications  -   Special Issue on Advances in Multicarrier
CDMA}, vol.  24, no.  6, pp.  1115  -   1129, June 2006.


\bibitem{Meshkati3} F. Meshkati, D. Guo, H. V. Poor and S. C. Schwartz, ``A unified
approach to energy-efficient power control in large CDMA systems,"
{\em IEEE Transactions on Wireless Communications}, vol.7, no.4,
p.p.1208 - 1216, April, 2008.


\bibitem{Meshkati4} F. Meshkati, H. V. Poor, S. C. Schwartz and R. Balan,
``Energy-efficient resource allocation in wireless networks with
quality-of-service constraints," {\em IEEE Transactions on
Communications},  to appear.

\bibitem{bib:Aazhang1}
A. Sendonaris, E. Erkip and B. Aazhang, ``User cooperation
diversity, Part I: System description,'' {\it IEEE Transactions on
Communications}, vol. 51, no. 11, pp. 1927 -  1938, November 2003.

\bibitem{bib:Laneman2}
J. N. Laneman, D. N. C. Tse and G. W. Wornell, ``Cooperative
diversity in wireless networks: Efficient protocols and outage
behavior,'' {\it IEEE Transactions on Information Theory}, vol.
50, no. 12, pp. 3062 -  3080, December 2004.

\bibitem{Yates2} I. Maric and R. D. Yates,
``Cooperative multihop broadcast for wireless networks," {\em IEEE
Journal on Selected Areas in Communications}, vol. 22, no. 6,
pp. 1080 -  1088, August 2004.

\bibitem{Luo} J. Luo, R. S. Blum, L. J. Greenstein, L. J. Cimini, and A. M. Haimovich,
``New approaches for cooperative use of multiple antennas in ad
hoc wireless networks," in {\em Proceedings of the IEEE Vehicular
Technology Conference}, vol. 4, pp. 2769 -   2773, Los Angeles, CA,
September 2004.

\bibitem{Bletsas} A. Bletsas, A. Lippman, and D. P. Reed,
``A simple distributed method for relay selection in cooperative
diversity wireless networks, based on reciprocity and channel
measurements," in {\em Proceedings of the IEEE Vehicular Technology
Conference}, vol. 3, pp. 1484 -  1488, Stockholm, Sweden, May 2005.

\bibitem{bib:zhuAhmed} A. K. Sadek, Z. Han, and K. J. R. Liu,
``An efficient cooperation protocol to extend coverage area in
cellular networks," in {\em Proceedings of the IEEE Wireless
Communications and Networking Conference}, vol. 3, pp. 1687 -  1692,
Las Vegas, NV, April 2006.

\bibitem{bib:zhuwhohelpswhom} Z. Han, T. Himsoon, W. Siriwongpairat,
and K. J. R. Liu, ``Energy efficient cooperative transmission over
multiuser OFDM networks: who helps whom and how to cooperate," in
{\em Proceedings of the IEEE Wireless Communications and Networking
Conference}, vol. 2, pp. 1030 -  1035, New Orleans, LA, March 2005.

\bibitem{globecom_zhu} B. Wang, Z. Han, and K. J. Ray Liu, ``Stackelberg game for
distributed resource allocation over multiuser cooperative
communication networks," in {\em Proceedings of the IEEE Global Telecommunications
Conference}, San Francisco, CA, November 2006.

\bibitem{hanzhu_CB}Z. Han and H. V. Poor, ``Lifetime improvement of wireless
sensor networks by collaborative beamforming and cooperative
transmission," in {\em Proceedings of the  IEEE International
Conference on Communications}, Glasgow, Scotland, June 2007.

\bibitem{bib:Adve}
Y. Zhao, R. S. Adve, and T. J. Lim, ``Improving
amplify-and-forward relay networks: Optimal power allocation
versus selection," in \emph{Proceedings of the IEEE International
Symposium on Information Theory}, Seattle, WA, July 2006.

\bibitem{bib:Madsen}
Z. Yang, J. Liu, and A. Host-Madsen, ``Cooperative routing and
power allocation in ad-hoc networks,'' in {\it Proceedings of the
IEEE Global Telecommunications Conference}, Dallas, TX, November
2005.

\bibitem{non_optimal}
S. Boyd and L. Vandenberghe, {\em Convex Optimization}, Cambridge
University Press, Cambridge, UK, 2006. (http://www.stanford.edu/ \~{}
boyd/cvxbook.html)


\bibitem{optimization2}Z. Han and K. J. R. Liu, {\em Resource Allocation for Wireless
Networks: Basics, Techniques, and Applications}, {Cambridge
University Press}, Cambridge, UK, 2008.

\end{thebibliography}

\newpage

\begin{figure}
    \centering
    \includegraphics[width=4in]{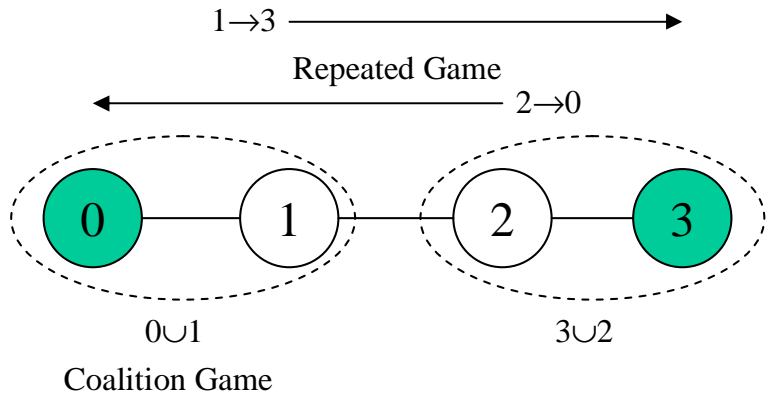}\vspace{5mm}
    \caption{\footnotesize{Illustration Example of the Curse of Boundary Nodes}}
    \label{example}
\end{figure}

\begin{figure}
    \centering
    \includegraphics[width=4in]{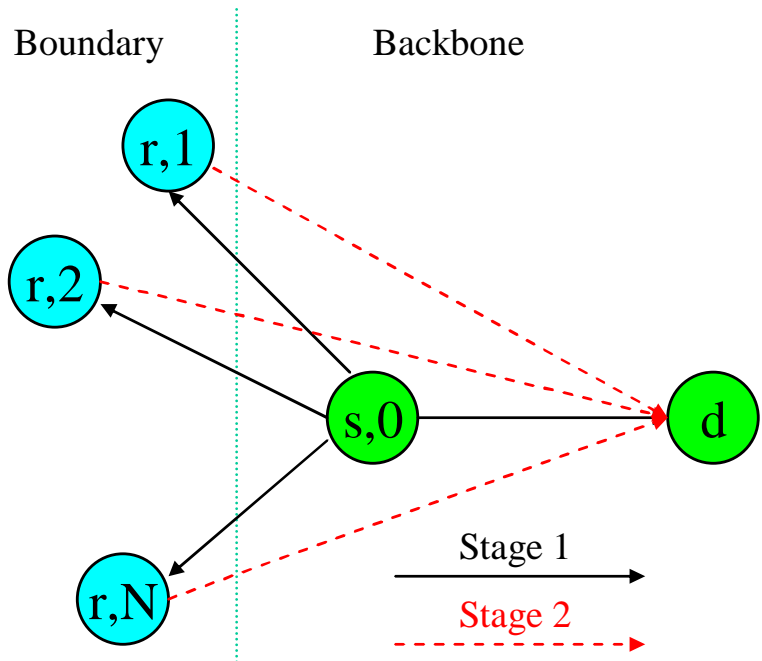}
    \caption{\footnotesize{Coalition Game with Cooperative Transmission}}
    \label{example_ct}
\end{figure}

\begin{figure}
    \centering
    \includegraphics[width=4in]{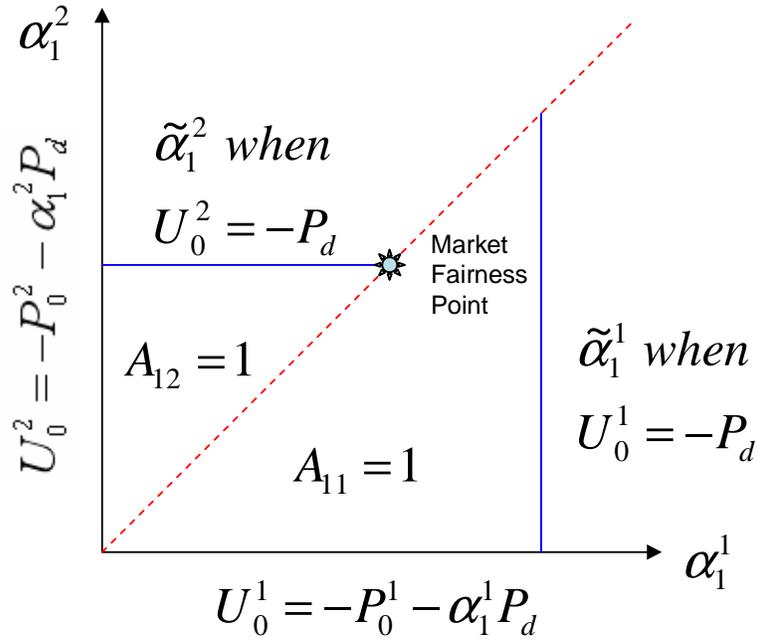}
    \caption{\footnotesize{Market Fairness Point for a Two-Backbone-Node One-Boundary-Node Example}}
    \label{example_mf}
\end{figure}

\begin{figure}
    \centering
    \includegraphics[width=4in]{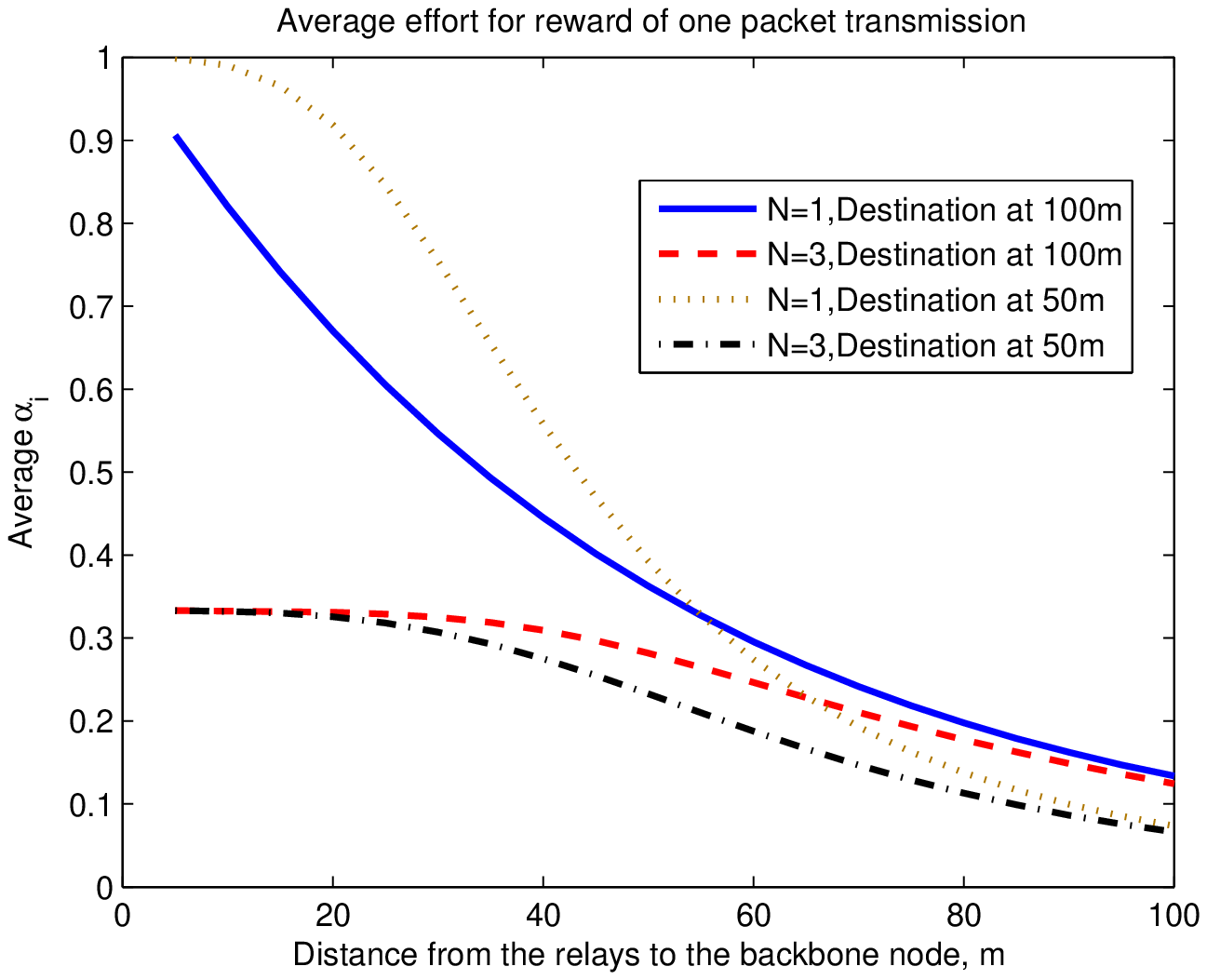}
    \caption{\footnotesize{$\alpha$ for Different Channels and No. of Nodes, Min-Max Fairness}}
    \label{a_R}
\end{figure}

\begin{figure}
    \centering
    \includegraphics[width=4in]{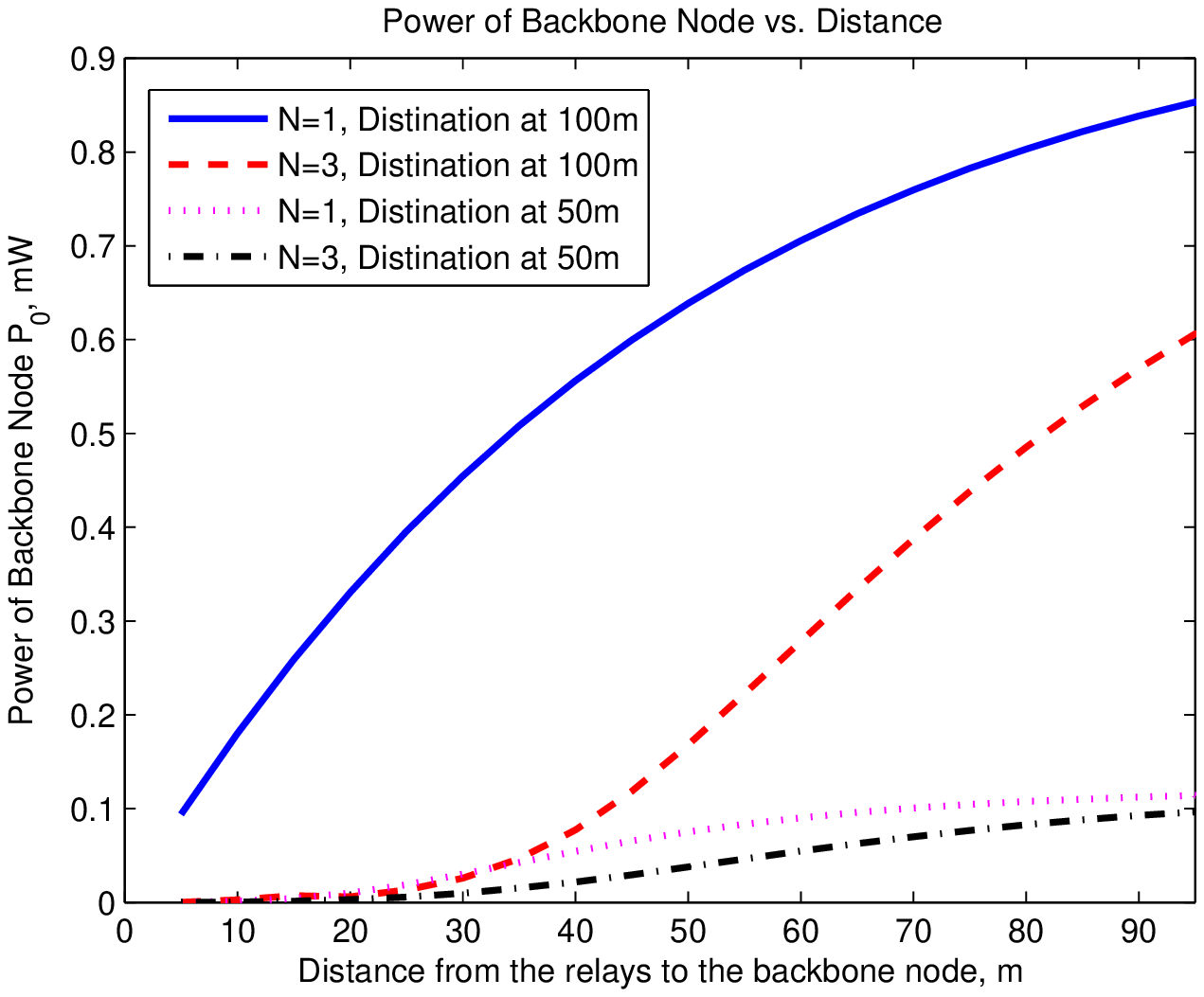}
    \caption{\footnotesize{$P_0$ for Different Channels and No. of Nodes, Min-Max Fairness}}
    \label{P0_R}
\end{figure}

\begin{figure}
    \centering
    \includegraphics[width=4in]{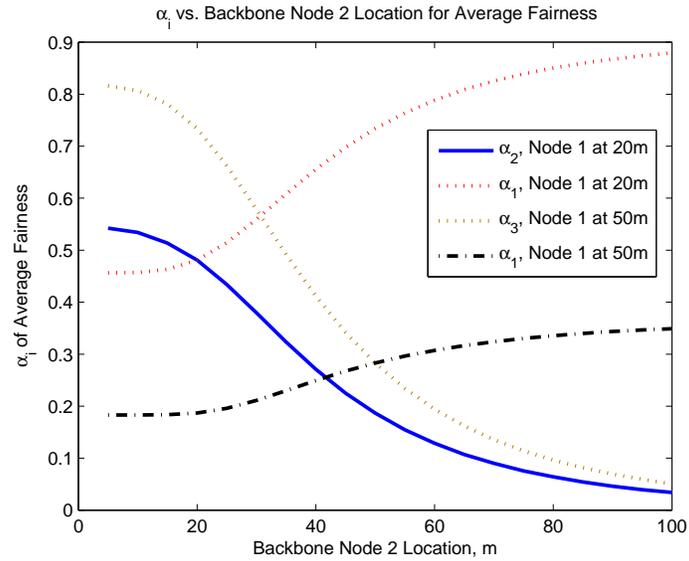}
    \caption{\footnotesize{$\alpha_i$ of Average Fairness for Different Users' Locations}}
    \label{a_shapley}
\end{figure}

\begin{figure}
    \centering
    \includegraphics[width=4in]{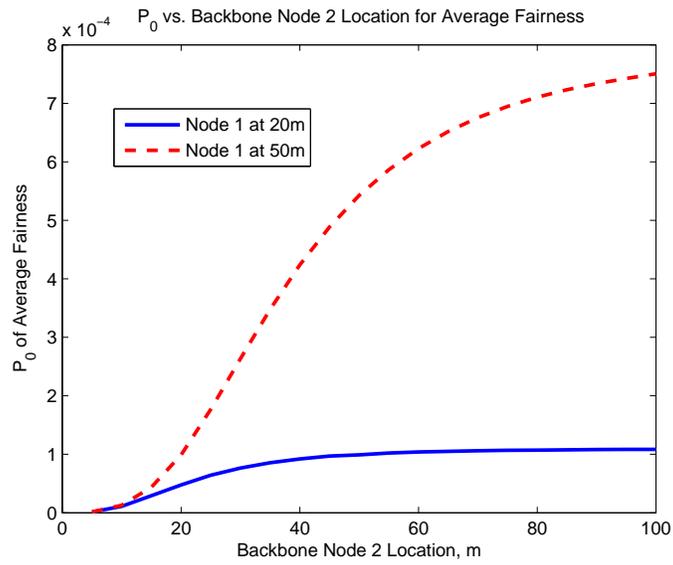}
    \caption{\footnotesize{$P_0$ of Average Fairness for Different Users' Locations}}
    \label{P0_shapley}
\end{figure}

\begin{figure}
    \centering
    \includegraphics[width=4in]{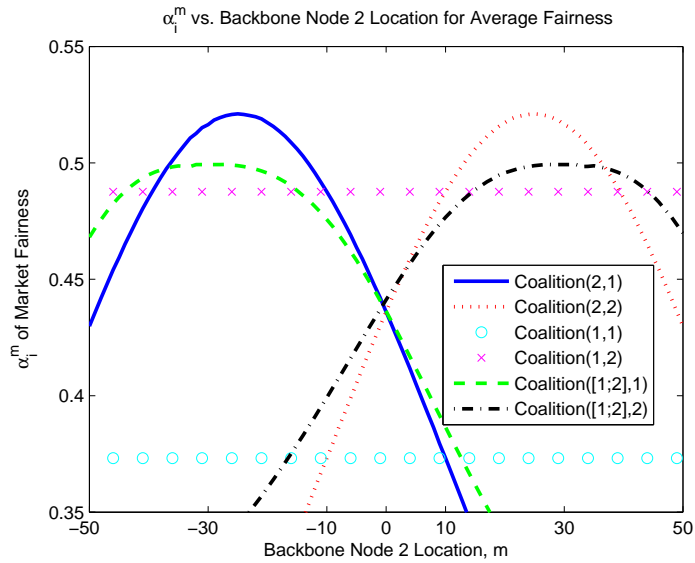}
    \caption{\footnotesize{$\alpha_i^m$ of Market Fairness for Different Users' Locations}}
    \label{a_market}
\end{figure}

\begin{figure}
    \centering
    \includegraphics[width=4in]{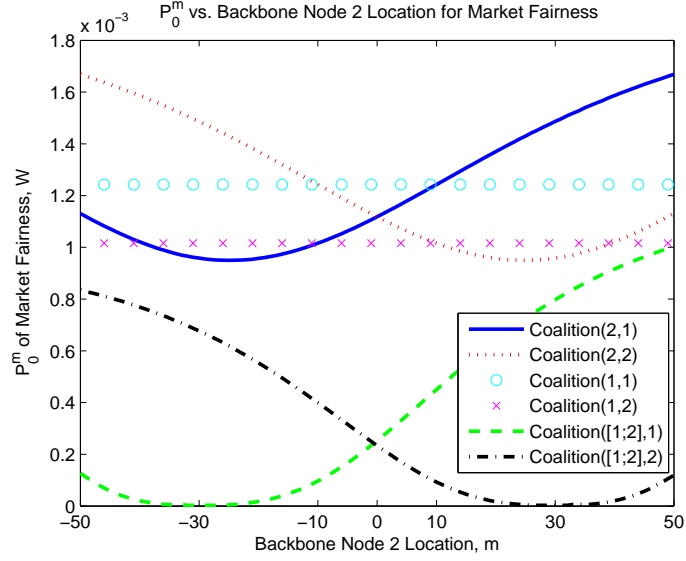}
    \caption{\footnotesize{$P_0^m$ of Market Fairness for Different Users' Locations}}
    \label{P0_market}
\end{figure}

\begin{table}
\caption{Market Fairness Coalitions}\vspace{3mm}
\begin{center}
\begin{tabular}{|c|c|c|}
  \hline
  % after \\: \hline or \cline{col1 -  col2} \cline{col3 -  col4} ...
  Case & Coalition (optimal for boundary node) & Minimal $\alpha$ offered by backbone node\\
  \hline
  I & (1,2),(2,1) & $\alpha_1^1,\alpha_2^2$ \\
  \hline
  II & ([1;2],1) &  $\alpha_1^2,\alpha_2^2$ \\
  \hline
  III & (2,1),(1,2) & $\alpha_1^1,\alpha_2^2$ \\
  \hline
  IV & (1,2),(2,1) & $\alpha_1^1,\alpha_2^2$ \\
  \hline
  V & (2,2),(1,1) & $\alpha_1^2,\alpha_2^1$ \\
  \hline
  VI & ([1;2],2) &  $\alpha_1^1,\alpha_2^1$ \\
  \hline
  VII & (1,2),(2,1) & $\alpha_1^1,\alpha_2^2$ \\
    \hline
\end{tabular}\label{market_table}
\end{center}
\end{table}

\begin{figure}
    \centering
    \includegraphics[width=4in]{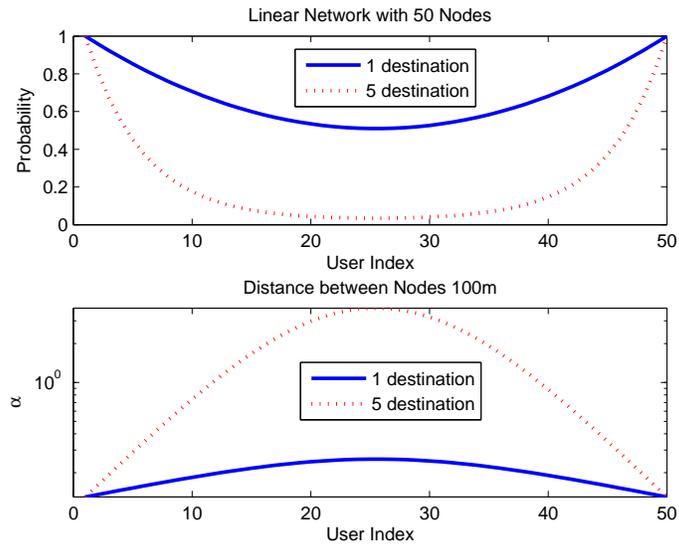}
    \caption{\footnotesize{Linear Network Setup}}
    \label{linear_setup}
\end{figure}

\begin{figure}
    \centering
    \includegraphics[width=4in]{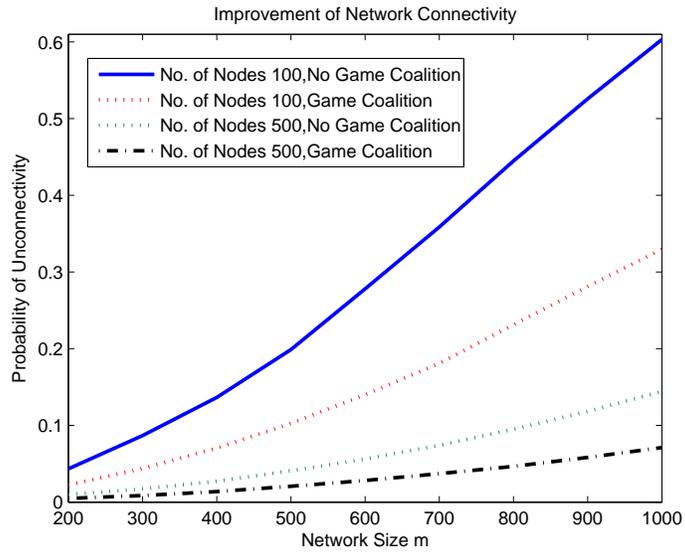}
    \caption{\footnotesize{Network Connectivity vs. Network Size}}
    \label{prob_size}
\end{figure}

\end{document}